\documentclass[twoside]{article}
\newcommand{\mytitle}{Hyperpfaffian Correlations for Beta-Ensembles: Beta an Even Square Integer} 
\newcommand{\keywords}{Random matrices, log-gases, beta ensembles, exterior algebra, correlation functions, pfaffians, hyperpfaffians, hyperpfaffian evaluations}
\newcommand{\msc}{ 15B52, 
60B20, 
60G55,  
82B23, 
15A15 
} 
\usepackage{amsmath,amssymb,
  amsthm,amsfonts,amscd,mathrsfs,pifont} 
\usepackage{eucal}  
\usepackage{verbatim}
\usepackage{ifpdf}
\usepackage[colorlinks=true, citecolor=blue]{hyperref}

\ifpdf
        \usepackage[pdftex]{graphicx}
\else
        \usepackage[dvips]{graphicx}
\fi

\newtheorem{thm}{Theorem}[section]
\newtheorem{cor}[thm]{Corollary}

\newtheorem{lemma}[thm]{Lemma}
\newtheorem{prop}[thm]{Proposition}

\newtheorem{thm*}{Theorem}[]
\newtheorem{cor*}[thm*]{Corollary}
\newtheorem{claim*}[thm*]{Claim}
\newtheorem{lemma*}[thm*]{Lemma}
\newtheorem{prop*}[thm*]{Proposition}
\newtheorem{conj*}[thm*]{Conjecture}

\theoremstyle{definition}

\newtheorem{problem*}{Problem}[section]

\newtheorem{question*}{Question}[section]

\newtheorem{defn*}{Definition}

\theoremstyle{remark}

\newcommand{\qq}[1]{\qquad \mbox{#1} \qquad}

\newcommand{\BB}[1]{{\mathbb{#1}}}

\newcommand{\R}{{\BB{R}}}
\newcommand{\Z}{{\BB{Z}}}
\newcommand{\C}{{\BB{C}}}
\newcommand{\T}{{\BB{T}}}

\newcommand{\bs}{\boldsymbol}
\newcommand{\mf}{\mathfrak}

\newcommand{\wt}{\widetilde}

\newcommand{\la}{\langle}
\newcommand{\ra}{\rangle}

\newcommand{\transpose}{{\mathsf{T}}}

\newcommand{\piecewise}[1]{
\left\{
\begin{array}{ll}
#1
\end{array}
\right.
}

\DeclareMathOperator{\sgn}{sgn}

\DeclareMathOperator{\Pf}{Pf}
\DeclareMathOperator{\PF}{PF}

\usepackage[margin=1 in]{geometry}
\allowdisplaybreaks

\numberwithin{equation}{section} 

\numberwithin{equation}{section}

\begin{document} 
\title{\bfseries\sffamily \mytitle}  
\author{{\sc Christopher D.~Sinclair} and {\sc Jonathan M.~Wells}}
\maketitle
  
\begin{abstract}
We give a hyperpfaffian formulation for correlation functions in
$\beta$-ensembles arising in random matrix theory and statistical mechanics when
$\beta = L^2$ is an even square integer. More specifically, for ensembles of
$M$ points in a space $W \subset \C$ (typically $W=\R$ or $W=\T$), arising either as
eigenvalues of a random matrix or as a system of charged particles with log
interaction, to the $m$th correlation function $R_m : W^m \rightarrow [0,
\infty)$ we associate the $L$-vector valued function $\gamma_m : W^m
\rightarrow \Lambda^L \C^{L(M-m)}$ such that $R_m(\mathbf y)$ is given by the
Vandermonde determinant in $y_1, \ldots, y_m$ times the hyperpfaffian of
$\gamma_m(\mathbf y).$  The partition function of the ensemble was previously shown to be
the hyperpfaffian of a {\it Gram} $L$-form $\gamma$ in $\Lambda^L \C^{LM},$ and
we demonstrate the relationship between $\gamma_m(\mathbf y)$ and $\gamma$, both
having coefficients built from integrals of Wronskians of monic polynomials.
Assuming the existence of families of polynomials sympathetic with the weight of
the ensemble, we may construct $\gamma(\mathbf y)$ so it is very sparse
(relative to the expected ${L(M-m) \choose L}$ coefficients of a general
$L$-vector). These generalize skew-orthogonal polynomials arising in the
well-understood $\beta = 4$ situation. Finally we explore the situation in the
circular $\beta = L^2$ ensembles. Here the monomials give a prototype, and we
give explicit formulas for $\gamma$ and $\gamma_m$ in this setting.
We use our hyperpfaffian framework to produce exact formulas for the two point
function when $\beta = 16$ for small values $M.$ Along the way we will record
hyperpfaffian evaluations using known values of partition functions of
$\beta$-ensembles. 
\end{abstract} 

{\bf MSC2020:} \msc

\vspace{1cm}

\section{Introduction}

The $\beta$-ensembles are a class of random point processes finding application in
random matrix theory, nuclear physics, and statistical mechanics (among several
other fields of study), as summarized briefly below. This class of processes
takes a common form indexed by a non-negative real parameter $\beta$. The joint
probability density for $M$ points $\mathbf x = (x_1, \dots, x_M)$ on a space
$W^M$ is specified by
\begin{equation}
	\label{eq:beta_en}
\frac{1}{Z(\beta)} \prod_{m < n}^M |x_n - x_m|^{\beta} \times \prod_{m=1}^M w(x_m)
\end{equation}
where $Z(\beta)$ is the {\em partition function} and $w$ is a prescribed {\em
weight}; a rigorous definition of this model appears in Sections
\ref{subsec:point-process} and \ref{subsec:beta-ensemble}.

The classic {\em Gaussian} (GOE, GUE, GSE) and {\em Circular} (COE, CUE, CSE)
$\beta$-ensembles are given by $\beta = 1, 2, 4$ with $w(x) = e^{-x^2/2}$
supported on the real line $W = \R$ (in the Gaussian case) and with $w(x) =
\frac{1}{2\pi}$ supported on the complex unit circle $W = \T$ (in the circular
case); they describe the joint distribution of eigenvalues of random matrices
with Gaussian entries invariant under the actions of orthogonal, unitary, or
symplectic matrices. The distribution of these eigenvalues is used to model
neutron resonance spacings in atomic nuclei, as proposed by Wigner
\cite{MR0095527} and further unified in an algebraic framework by Dyson
\cite{Dyson1962}.

When $\beta = 2$ the correlation functions (to be defined exactly below) can be
expressed as determinants of matrices formed from the reproducing kernel of the
weight \cite{MR2129906, MR1657844}. This kernel can then be analysed as $M
\rightarrow \infty$ to understand statistical properties of various scaling
limits of the eigenvalues in this limit \cite{MR1844228, MR573370}. When $\beta
= 1$ and $4$ the correlation functions can be given as Pfaffians of
anti-symmetric matrices formed from a {\em matrix} kernel which behaves like a
skew-symmetric reproducing kernel. Similar analyses of these matrix kernels
allow us to understand scaling limits of the eigenvalues of $\beta =1, 4$
ensembles \cite{MR1385083, MR1675356}. In all cases, it was the observation that
the partition functions are Gram determinants or (antisymmetric Gram) Pfaffians
that begin the derivation of the determinantal or Pfaffian correlations.
Determinantal and Pfaffian point processes \cite{MR1799012, MR2013698,
MR2013699} are central objects in the study of random matrix theory and point
processes. (See, for instance \cite{MR2129906} for the early development of
random matrix theory).

Later, it was demonstrated that there are matrix ensembles for all $\beta \geq
0$, though the structure of the matrix entries is quite different than for the
classic matrix ensembles \cite{MR1936554}. The special structure of these random
matrices allows for some analysis of eigenvalue statistics as $M \rightarrow
\infty,$ \cite{MR2331033, MR2813333} though, aside from $\beta=1,2,4$,
determinantal or Pfaffian correlations are lacking (and arise in the classic
ensembles in ways that will not naturally generalize to non-integer $\beta$).

An alternative perspective on $\beta$-ensembles is afforded by the log-gas
model, first identified by Dyson \cite{Dyson1962}. The relationship between the
log-gas perspective and eigenvalues of random matrices is summarized in
\cite{Forrester2015} and explored in great detail in \cite{Forrester2010}. In
this model, a system of two-dimensional charged particles in thermal equilibrium
is constrained to lie along a region of the plane with prescribed background
charge density, where particles have pairwise interaction determined by
logarithmic distance. The distribution of particle positions $\bf x$ is
proportional to the Boltzmann factor $e^{ - \beta U(\bf x) }$, where $U(\bf x)$
describes the total potential energy of the system and $\beta \geq 0$ represents
the inverse temperature. Here, we assume that $U$ decomposes as the sum of
one-body potentials $V(x_i)$ and the pairwise interactions $-\log |x_i - x_j|$,
and hence, the Boltzmann factor is proportional to the form given in
$\ref{eq:beta_en}$, with $w(x) = e^{-\beta V(x)}$.

Regardless of what process the $\beta$-ensemble models, for certain integer
values of $\beta$, the partition function admits a hyperpfaffian formulation
\cite{MR1943369, springerlink:10.1007/s00605-011-0371-8, MR4035418}, and it is a
subset of these $\beta$, when $\beta = L^2$ is an even integer, that we will
consider here. Hyperpfaffians are generalizations of Pfaffians, but instead of
acting on anti-symmetric matrices, they act on multivectors (alternating
tensors). The Gram matrices which appear in the classic ensembles will be
replaced with Gram $L$-vectors in the $\beta = L^2$ situation, and their
hyperpfaffians yield the partition function. This is a fairly tidy
generalization of the $\beta = 4$ situation to that of all $\beta = L^2$ even.
There is also a generalization of the $\beta = 1$ situation to the $\beta = L^2$
odd case, but the odd case is more nuanced and we will leave it for the future.

The existence of hyperpfaffian partition functions suggests the existence of
hyperpfaffian correlations, and we will use the `averaged characteristic
polynomial' trick to derive such hyperpfaffian correlations when $\beta = L^2$
is even. This marks the first step towards our complete goal of defining a
suitable generalization of Pfaffian point processes for these ensembles, which
will require further investigation since hyperpfaffian correlations are not
formed from a kernel in the same manner as the $\beta = 1, 4$ case; the
obstruction to suitable generalization is discussed further in
Subsection~\ref{subsec:future}. However, our methods are exact and given in a
form that may be amenable to induction on the number of particles.

As is usually the case in random matrix theory, the circular ensembles are more
readily tractable as compared to Hermitian ensembles \cite{MR177643, MR0143556,
MR0143557, MR0143558}. The same seems to be true here, and we will invest
considerable time looking at the $\beta = L^2$ even circular ensembles. In this
situation, which we hope is generalizable to the Hermitian case, we can
explicitly produce the $L$-vectors whose hyperpfaffian yields the $m$th
correlation function, and for small values of $L$ and $M$ and $m$ compute these
hyperpfaffians.

The remainder of the paper is organized as follows: Section~\ref{sec:beta}
introduces the formal probabilistic framework for our point process models, and
further discusses the physical interpretation of $\beta$ in the log-gas models;
we also discuss potential generalizations to the multicomponent model.
Section~\ref{sec:alg} outlines the algebraic and combinatoric structures in the
exterior algebra that are used in our proofs. Section~\ref{sec:results} provides
the main result of the paper: an exact hyperpfaffian formula for the correlation
functions of the $\beta$-ensembles when $\beta = L^2$ is an even integer, and
its specialization to the circular $\beta$-ensembles. Proofs of these results
are contained in Section~\ref{sec:proofs}. Finally, the appendix collects
explicit formulas for the pair correlation functions for the circular ensembles
for small values of $M$ and $\beta$, expressed as polynomials in $\cos$ with
rational coefficients.

\section{$\beta$-Ensembles}
	\label{sec:beta}
\subsection{Point Processes}
\label{subsec:point-process}

Let $\mathbb F$ be a complete field (usually $\R$ or $\C$) with absolute value
$| \cdot |.$ Let $W \subset \mathbb F$ and suppose $(W, \mathcal B, \mu)$ is a
measure space. We denote by $\mu^M$ the product measure on the product
$\sigma$-algebra $\mathcal B^{\otimes M}$ of $W^M$. Given a set $B \in \mathcal
B$ we define the measurable function $N_B : W^M \rightarrow \mathbb N$ by 
\[
N_B(\mathbf x) = \# \{x_1, \ldots, x_M\} \cap B.
\]
That is $N_B(\mathbf x)$ gives the number of coordinates of $\mathbf x$ in $B$.
We define $\mathcal C \subset \mathcal B^{\otimes M}$ to be the {\em cylinder}
$\sigma$-algebra generated by all $N_B$, $\mathcal C = \sigma\{ N_B : B \in
\mathcal B\}$. An $M$ particle {\em point process} on $W$ is a probability space
$(W^M, \mathcal C, \mathbb P)$. 

A common way of defining a point process is to provide a joint distribution
$\nu$ for a random vector $\mathbf X = (X_1, \ldots, X_M) \in W^M$ which we view
as random locations of particles in $W$. Under this interpretation, $N_B(\mathbf
X)$ is the random number of particles in $B \subset W$. When $\nu$ is restricted
to $\mathcal C$, we lose the ability to distinguish {\it which} of the $X_1,
\ldots, X_M$ lie in a given set $B$, and only have access to {\it how many} are
in $B$. This situation is most applicable to that when the particles are
indistinguishable. The $X_1, \ldots, X_M$ are {\em exchangeable} if they have
the same distribution. If $f(\mathbf x)$ is the joint density of $\mathbf X$
with respect to $\mu^M$, then $X_1, \ldots, X_M$ are exchangeable if and only if
$f(\mathbf x)$ is ($\mu^M$-a.e.) invariant under permutation of the coordinates
of $\mathbf x$. 

Let $1 \leq m \leq M$. The function $R_m : W^m \rightarrow [0, \infty)$ is
called the $m$th correlation function of the ensemble if, given any pairwise
disjoint sets $B_1, \ldots, B_m \in \mathcal B$,
\[
\mathbb E[N_{B_1} \cdots N_{B_m}] = \int_{B_1} \cdots \int_{B_m} R_m(\mathbf y) \, d\mu^m(\mathbf y).
\]
The correlation functions, if they exist, characterize the point process. In a
sense this observation is trivial, because $R_M \propto f$, but the utility of
correlation functions is that, if our interest is the occupation numbers of $m$
disjoint sets, then we need only do $m$ integrations (as opposed to $M$
integrations if we appeal directly to the joint density). It is not difficult to
verify from definition that
\[
R_m(\mathbf y) = \frac{M!}{(M-m)!} \int_{W^{M-m}} f(y_1, \ldots, y_m, x_1, \ldots, x_{M-m}) \, d\mu^{M-m}(\mathbf x).
\]
That is, up to a combinatorial constant, the $m$th correlation function is the
$m$th marginal density (and because $X_1, \ldots, X_M$ are exchangeable, it does
not matter which $m$ random variables we look at the marginal density for). 

\subsection{$\beta$-Ensembles}
\label{subsec:beta-ensemble}

A $\beta$-ensemble on $W$ with weight function $w: W \rightarrow [0,\infty)$ is
a point process on $W$ specified by joint density on $W^M$ given by
\[
f(\mathbf x) = \frac{1}{M! Z} \prod_{m < n}^M |x_n - x_m|^{\beta} \cdot \prod_{\ell=1}^M w(x_{\ell}), \qq{where} Z = \frac{1}{M!}\int_{W^M} \prod_{m < n}^M |x_n - x_m|^{\beta} \cdot \prod_{\ell=1}^M w(x_{\ell}) \cdot d\mu^M(\mathbf x).
\]
$Z$ is called the {\em partition function} of the ensemble. 

We will restrict ourselves to $W \subset \mathbb R$ or $\C$ such that there
exists $c : W \rightarrow \C$ so that $| x - y |^2 = c(x) c(y) (x-y)^2$. This
may seem like a strange condition, but it allows us to unify the formal theory
for the circular ($W = \mathbb T \subset \C$) and real $\beta$-ensembles, when
$\beta = L^2$ is an even integer. If $W \subset \R$ we may set $c = 1$. When
$x,y \in \T$, 
\[
| x - y |^2 = (x - y)(\overline x - \overline y) = (x - y)\left(\frac1x - \frac1y\right) = -\frac{1}{x y}(x - y)^2,
\]
thus, for circular ensembles we set $c(x) = i/x$. It follows that if $\beta$ is
an even integer, 
\[
\prod_{m < n} |x_n - x_m|^{\beta} = \prod_{m < n} c^{\beta/2}(x_m) c^{\beta/2}(x_n) (x_n - x_m)^\beta = \prod_{m < n} (x_n - x_m)^\beta \cdot \prod_{j=1}^M c^{(M-1)\beta/2}(x_j).
\]
If we define $u(x) = c^{(M-1)\beta/2}(x) w(x)$, then we have the unified formula
for the joint density
\[
f(\mathbf x) = \frac{1}{M! Z} \prod_{m < n}^M (x_n - x_m)^{\beta} \cdot \prod_{\ell=1}^M u(x_{\ell}).
\]
In both cases we refer to $u(x)$ as the weight for the ensemble; when $W \subset
\R$ it is equal to the weight, and when $W = \T$, even though it is complex, it
formally plays the same role.

\subsection{Ensembles of Charged Particles}

The traditional interpretation of $\beta$ is the dimensionless {\em inverse
temperature} $1/kT$ where $k$ is Boltzmann's constant. Under this
interpretation, $\beta = 1$ ensembles represent systems of $M$ unit charged
particles at inverse temperature $1/kT = 1.$ Likewise, $\beta = 4$ ensembles
represent unit charged particles at inverse temperature $1/kT = 4.$ This
interpretation works for all $\beta > 0.$

When $\beta = L^2$ there is another equally valid interpretation. In this
interpretation we fix the inverse temperature $1/kT = 1,$ and we view the
particles as being identical of charge $L.$ This interpretation could be
extended to all $\beta > 0$ by allowing particles of charge $\sqrt{\beta},$
though this feels non-physical. Regardless, viewing the temperature as fixed and
the charges as varying with $\beta,$ allows us to construct more sophisticated
{\em multicomponent} ensembles where particles of different (positive) integer
charges interact. The partition functions of such ensembles are given by a
generalization of the hyperpfaffian known as the {\em Berezin Integral}
\cite{MR0208930, 1751-8121-45-16-165002, MR4480836}. We will not explore the
multicomponent situation here, but we expect there to be a Berezin integral
formulation for correlations in such multicomponent ensembles following a more
sophisticated analysis than we pursue here.  

\subsection{Pair Correlation in Circular $\beta$-Ensembles}

Here we preview some corollaries of our main result as applied to $M$-particle
circular $\beta$-ensembles when $\beta = 4, 16$ and $36$ for small values of
$M$. The $\beta=4$ case is classical, but our methods are applicable here and we
recover the expected pair correlation functions and compare them with their
$\beta=16$ counterparts. By our previous remarks, we may interpret this as $M$
particles on the circle with charge $2$ or $4$ (corresponding to $\beta=4$ and
$\beta=16$ respectively) at inverse temperature $1/kT = 1$. 

The joint density $f(\mathbf x)$ is invariant under simultaneous rotation of its
arguments, and by extension so are the correlation functions. This allows us to
reduce the number of variables necessary to describe the correlation functions
by one. In particular, we write $R_2(\theta) := R_2(e^{i \theta}, e^{-i\theta})$
for the second or {\em pair} correlation function. By way of intuition, given a
pair of particles in an $M$ particle ensemble, $R_2(\theta)$ should be largest
when $\theta$ is near a (non-integral) multiple of $\pi/M$---that is when the
gap between the particles is close to a multiple of $2 \pi/M$. When comparing
$\beta = 4$ and $\beta = 16$ we expect that while the maxima and minima of
$R_2(\theta)$ are both near multiples of $\pi/M,$ the larger charge for
$\beta=16$ suggests the maxima and minima of $R_2(\theta)$ will be more extreme
than the $\beta=4$ case. Moreover, small values of $\theta$ represent situations
where two particles are nearby, a situation much more unlikely when $\beta=16$
as compared to $\beta=4$. Thus the graph of the former should be `flatter' near
the origin when compared to the latter. Figures~\ref{fig:2pt} and \ref{fig:2pt2}
were generated using the methods described here. Represented are graphs of
$R_2(\theta)$ for $\beta=4$ and $\beta=16$ when $M=4,5,6.$ These are given
explicitly as polynomials in $\cos(\theta)$ with rational coefficients up to a
single factor of $\pi^{-1}.$ See the appendix for exact formulas.

To discover these polynomials we need to do some calculations in the exterior
algebra.

\begin{figure}
	\includegraphics[width=6in]{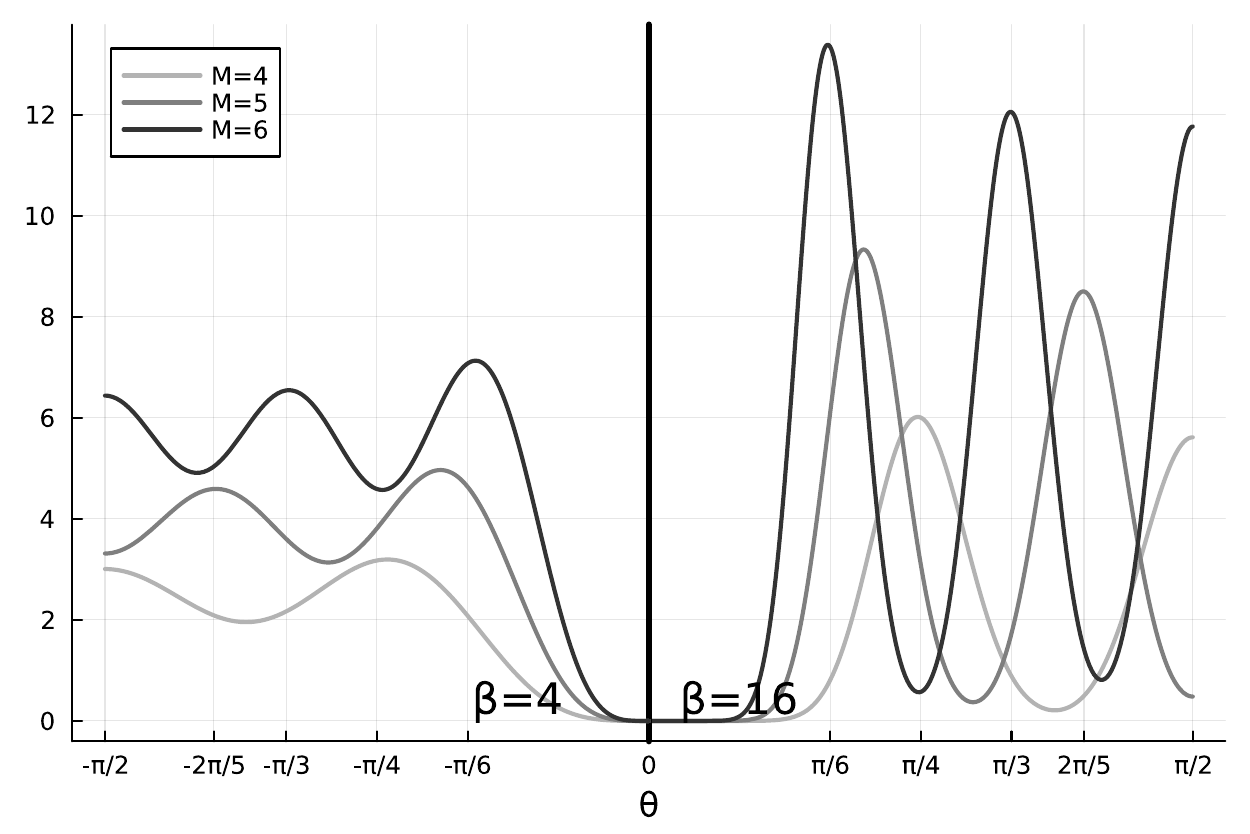}
	\caption{\label{fig:2pt}$R_2(\theta)$ for $\beta=4$ to the left and
	$\beta=16$ to the right and small values of $M$.}
\end{figure}

\begin{figure}
	\includegraphics[width=6in]{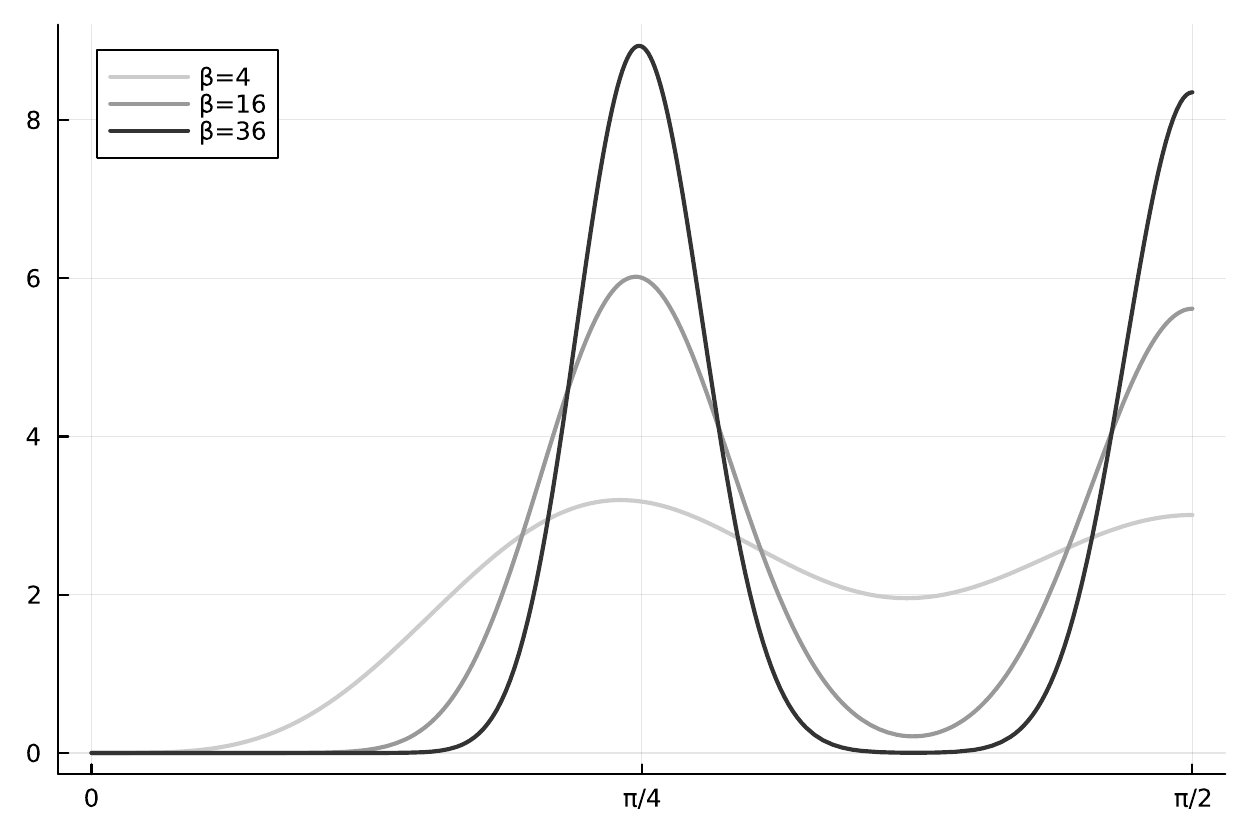}
	\caption{\label{fig:2pt2}$R_2(\theta)$ for $\beta=4, 16, 36$ and $M=4$.}
\end{figure}

\section{Pfaffians and Hyperpfaffians}
	\label{sec:alg}
\subsection{Index Notation}

Given a set $A$ and positive integer $j$, we define ${A \choose j}$ to be the
collection of subsets of $A$ of cardinality $j$,
\[
{A \choose j} := \{ B \subset A : \#B = j \}.
\]
Of course the cardinality of ${A \choose j}$ is ${\#A \choose j}$. The set of
$j$-tuples with coordinates in $A$ is denoted $A^j$. If $\#A = J < \infty$ and
$j_1, \ldots, j_M$ are non-negative integers such that $j_1 + \cdots + j_M = J$,
then we define ${A \choose j_1, \ldots, j_M}$ to be the ordered (set) partitions
of $A$ into pairwise disjoint sets of size $j_1, \ldots, j_M$. That is,
\[
{A \choose j_1, \ldots, j_M} := \bigg\{ \vec{\mf u} = (\mf u_1, \ldots, \mf u_M) : \mf u_m \in {A \choose j_m}, \mf u_m \cap \mf u_n = \emptyset \mbox{ for } m \neq n \bigg\}.
\]

Given a non-negative integer $J$ we define $[J] = \{0, 1, \ldots, J\}, [J) =
\{0, 1, \ldots, J-1\}, (J) = \{1, 2, \ldots, J-1\}$ and $(J] = \{1, 2, \ldots,
J\}.$ Given $0 \leq j \leq J$ and  $\mf u \in {(J] \choose j}$ we denote the
elements of $\mf u$ by $\{\mf u(1), \ldots, \mf u(j)\}$ ordered so that $0 < \mf
u(1) < \mf u(2) < \ldots < \mf u(j)\leq J$. We define $\mf u' \in {(J] \choose
J-j}$ to be the complement of $\mf u$ in $(J]$; $\mf u' = (J] \setminus \mf u$.
It is sometimes useful to view elements of $(J]^j$ and ${(J] \choose j}$ as
functions from $(j] \rightarrow (J]$. Viewed as a function, $\mf u \in {(J]
\choose j}$ is a strictly increasing function $(j] \nearrow (J]$, and $\mf u'$ is
the unique increasing function $(J-j] \nearrow (J]$ whose range is disjoint from
$\mf u$.

\subsection{The Exterior Algebra}

Let $V$ be a vector space of dimension $N$ with basis $\mathbf e_1, \ldots,
\mathbf e_N$ over a field $\mathbb F$ (this can be over $\R$ or $\C$, or some
other field, depending on context). The exterior algebra over $V$, $\Lambda V$
is the algebra with product denoted $\wedge$, generated by the relations $\{
\mathbf e_n \wedge \mathbf e_m = - \mathbf e_m \wedge \mathbf e_n : m, n \in
(N]\}$. The exterior algebra is graded, 
\[
\Lambda V = \bigoplus_{n=0}^N \Lambda^n V, 
\]
where $\Lambda^n V$ is the $\mathbb F$-vector space with basis $\{ \mathbf
e_{\mf t} := \mathbf e_{\mf t(1)} \wedge \cdots \wedge \mathbf e_{\mf t(n)} :
\mf t \in {(N] \choose n} \}$. Elements of $\Lambda^n V$ are known as
$n$-vectors, and those of the form $\mathbf v_1 \wedge \cdots \wedge \mathbf
v_n$ for linearly independent $\mathbf v_1, \ldots, \mathbf v_n \in V$ are
called $n$-blades. It can be difficult to determine whether a given $n$-vector
is an $n$-blade. If $\alpha \in \Lambda^n V$ we write $\alpha_{\mf t}$ for the
coordinate of $\alpha$ with respect to the basis element $\mathbf e_{\mf t}$.
That is,
\[
\alpha = \sum_{\mf t \in {(N] \choose n}} \alpha_{\mf t} \mathbf e_{\mf t}.
\]
$\Lambda^n V$ has dimension ${N \choose n}$. In particular $\Lambda^N V$ is the
one-dimensional {\em determinantal line} with basis element $\mathbf e_{(N]} =
\mathbf e_1 \wedge \cdots \wedge \mathbf e_N$ known as a {\em volume form}. 

Given $\mf t \in {(N] \choose n}$, we define $\sgn(\mf t) \in \{-1,1\}$ by
$\mathbf e_{\mf t} \wedge \mathbf e_{\mf t'} = \sgn(\mf t) \mathbf e_{(N]}$. The
{\em Hodge star} operator is an isomorphism on $\Lambda V$ which maps $\Lambda^n
V$ to $\Lambda^{N-n} V$ given on a basis by $\ast \mathbf e_{\mf t'} = \sgn(\mf
t) \mathbf e_{\mf t}$. Note that the Hodge star maps the determinantal line to
$\Lambda^0 V = \mathbb F$. If $\vec{\mf u} = (\mf u_1, \ldots, \mf u_M)$ is a
partition of $(N]$, then $\sgn(\vec{\mf u}) \in \{-1, 1\}$ is defined by
\[
\mathbf e_{\mf u_1} \wedge \cdots \wedge \mathbf e_{\mf u_M} = \sgn(\vec{\mf u}) \, \mathbf e_{(N]}.
\]

Let $V^\ast$ be the dual of $V$, and consider the pairings between $\Lambda^n
V^\ast$ and $\Lambda^n V$ given by
\[
[ \mathbf b_1 \wedge \cdots \wedge \mathbf b_n, \mathbf a_1 \wedge \cdots \wedge \mathbf a_n] = \det \left[ \mathbf b_j(\mathbf a_k)\right]_{j,k=1}^n.
\]
Given a multivector $\alpha \in \Lambda^{N-n} V$ we get a linear functional on
$\Lambda^n V$ by $\gamma \mapsto \ast \alpha \wedge \gamma,$ and it follows from
dimension considerations that $(\Lambda^n V)^\ast$ is isomorphic to
$\Lambda^{N-n} V.$ The Hodge star thus induces an isomorphism from $\Lambda^n V$
and $\Lambda^n V^\ast$ by $\alpha \mapsto (\gamma \mapsto \ast (\ast \alpha)
\wedge \gamma ).$ (Note the order of operations: we apply the Hodge star after
wedge products). If $\mf t, \mf u \in {(N] \choose L}$ then $[\ast \mathbf
e_{\mf t}, \mathbf e_{\mf u}] = \delta_{\mf t, \mf u}$ (the Kronecker $\delta$).

$N \times N$ matrices act on $\Lambda^n V$ (and all of $\Lambda V$ by extension)
via the map 
\[
\mathbf B \cdot \mathbf e_{\mf t} = \mathbf B \mathbf e_{\mf t(1)} \wedge \cdots \wedge \mathbf B \mathbf e_{\mf t(n)}.
\]
If we look at the pairing under this action, we find $\mathbf b_j(\mathbf B
\mathbf a_k) = (\mathbf B^\transpose \mathbf b_j)(\mathbf a_k)$ and hence
\[
[ \mathbf b_1 \wedge \cdots \wedge \mathbf b_n, \mathbf B \cdot \mathbf a_1 \wedge \cdots \wedge \mathbf a_n] = [ \mathbf B^\transpose \cdot \mathbf b_1 \wedge \cdots \wedge \mathbf b_n, \mathbf a_1 \wedge \cdots \wedge \mathbf a_n].
\]

\subsection{Pfaffians and Hyperpfaffians}

Let $\mathbf A = [a_{n,m}]$ be an antisymmetric $2M \times 2M$ matrix. As a
polynomial in the entries of $\mathbf A$, $\det \mathbf A$ is the square of a
polynomial of half the degree. This polynomial is known as the {\em Pfaffian} of
$\mathbf A$, and it can be explicitly given as a sum over the symmetric group
$S_{2M}$ by
\[
\Pf(\mathbf A) = \frac{1}{2^M M!} \sum_{\sigma \in S_{2M}} \sgn(\sigma) \prod_{m=1}^M a_{\sigma(2m-1), \sigma(2m)}.
\]
A couple important (and easy to compute) $2M \times 2M$ examples are
\begin{equation}
	\label{eq:pfaffs}
	\Pf \begin{bmatrix} 
		0    & c_1  &        & &\\          
		-c_1 & 0    &        & &\\
		&      & \ddots & &\\
		&      &        & 0   & c_M \\
		&      &        & -c_M &  0
	\end{bmatrix} = \prod_{m=1}^M c_m \qq{and} 
	\Pf \begin{bmatrix}
		0  & 1 & 1 & 1 & \\
		-1 & 0 & 1 & 1 & \\
		-1 & -1 & 0 & 1 & \\
		-1 & -1 & -1 & 0 & \\
		&    &    &   & \ddots \\
	\end{bmatrix} = 1.
\end{equation}
The Pfaffian has an equivalent definition in the exterior algebra. Associated to
the antisymmetric matrix is the 2-vector $\alpha \in \Lambda^2 \mathbb F^{2M}$
with coefficient for $\mathbf e_{\mf t}$ given by $\alpha_{\mf t} = a_{\mf t(1),
\mf t(2)}$. That is,
\[
\alpha = \sum_{m < n}^{2M} a_{m, n} \mathbf e_{m} \wedge \mathbf e_{n} = \sum_{\mf t \in {(2M] \choose 2}} \alpha_{\mf t} \mathbf e_{\mf t}.
\]
It follows then that $\alpha^{\wedge M}/M!$ is on the determinantal line, and
the coefficient of $\mathbf e_{(2M]}$ turns out to be the Pfaffian of $\mathbf
A$. That is,
\[
\Pf(\mathbf A) = \ast \frac{\alpha^{\wedge M}}{M!}.
\]
We will write $\PF(\alpha)$ for this number and call it the Pfaffian of
$\alpha$. Our previous examples can be reexpressed as
\[
\PF\left( \sum_{m=1}^M c_m \mathbf e_{2m-1} \wedge \mathbf e_{2m} \right) = \prod_{m=1}^M c_m \qq{and} \PF\left(\sum_{\mf t \in {(2M] \choose 2}}  \mathbf e_{\mf t} \right) = 1.
\]
When viewed in the exterior algebra, certain properties of the Pfaffian become
obvious. For instance, suppose $\mathbf f_1, \ldots, \mathbf f_{2M}$ is another
basis for $\mathbb F^{2M}$ with change of basis matrix $\mathbf B$ given by
$\mathbf f_m = \mathbf B \mathbf e_m$. Then $\mathbf f_1 \wedge \mathbf f_2
\wedge \cdots \wedge \mathbf f_{2M}$ is on the determinantal line, and its
coordinate with respect to $\mathbf e_{(2M]}$ is the determinant of $\mathbf B$.
That is,
\[
\det(\mathbf B) = \ast \big( \mathbf B \mathbf e_1 \wedge \mathbf B \mathbf e_2 \wedge \cdots \wedge \mathbf B \mathbf e_{2M} \big).
\]
The multiplicativity of the determinant follows trivially from this.  If we
write $\alpha \in \Lambda^2 \mathbb F^{2M}$ in our new basis, then 
\[
\mathbf f_m \wedge \mathbf f_n = \mathbf B \mathbf e_m \wedge \mathbf B \mathbf e_n,
\]
and if we want to write $\mathbf A$ in the new basis, then this produces the
matrix $\mathbf B^T \mathbf A \mathbf B$. It follows that if we compute the
Pfaffian of $\alpha$ in the new basis, using $\mathbf f_{[2M)}$ for the volume
form, then this changes the Pfaffian by $\det(\mathbf B)$. That is $\Pf(\mathbf
B^T \mathbf A \mathbf B) = \det(\mathbf B) \Pf(\mathbf A)$. This is the Pfaffian
analog of the multiplicativity of the determinant. 

When $L$ is an even integer, we can extend the exterior algebra definition to
$L$-vectors by noting that if $\omega \in \Lambda^L \mathbb F^{L M}$ then
$\omega^M/M!$ is on the determinantal line, and we can use the same definition
to write
\[
\PF(\omega) = \ast \frac{\omega^{\wedge M}}{M!}
\]
for the {\em hyperpfaffian} of $\omega$. (When $L$ is odd, this definition
always produces 0.) 

\subsection{Confluent Vandermonde Determinants}

Let $L$ and $M$ be positive integers and $V$ a vector space of dimension $N =
LM.$ It will be convenient to index our preferred basis for $V$ starting at $0.$
That is $V = \mathrm{span}_{\mathbb F}\{\mathbf e_0, \ldots, \mathbf e_{N-1}\}.$
The induced basis for $\Lambda^n V$ is given by $\{\mathbf e_{\mf t} : \mf t \in
{[N) \choose n} \}.$ 

Let $\mathbf p(x) = (p_0(x), p_1(x), \ldots, p_{N - 1}(x))$ be a vector of monic
polynomials in $\mathbb F[x]$ with $\deg p_n = n$. We say $\mathbf p : [LM]
\nearrow \mathbb F[x]$ is a {\em complete} family of monic polynomials. We
define the (modified) $\ell$th derivative operator by $D^{\ell} =
\frac{1}{\ell!} \frac{d^{\ell}}{dx^{\ell}}$ and write $D^{\ell} \mathbf p(x) =
\left( D^{\ell} p_n(x) \right)_{n=0}^{LN-1}$.  We define $\omega : \mathbb F
\rightarrow \Lambda^L V$ by 
\[
\omega(x) = \mathbf p(x) \wedge D^1 \mathbf p(x) \wedge \cdots \wedge D^{L-1} \mathbf p(x).
\]
The confluent Vandermonde determinant identity \cite{MR130257} then implies
\[
\ast  \omega(x_1) \wedge \omega(x_2) \wedge \cdots \wedge \omega(x_M)  = \prod_{m < n} (x_n - x_m)^{L^2}.
\]
Of importance is the fact that this determinant is independent of the complete
family of polynomials employed. 

\subsection{Wronskians}

We may write $\omega(x)$ with respect to our preferred basis $\{\mathbf e_{\mf t} :
\mf t \in {[N) \choose L}\}$ using coordinate functions $\omega_{\mf t} : W
\rightarrow \mathbb F,$
\[
\omega(x) = \sum_{\mf t \in {(N] \choose L}} \omega_{\mf t}(x) \mathbf e_{\mf t}.
\]
These are given by the {\em Grassmann coordinates}. The coordinate $\omega_{\mf
t}$ is given explicitly by the determinant of the $L \times L$ minor of
$[\mathbf p(x) \quad D^1 \mathbf p(x) \quad \cdots \quad  D^{L-1} \mathbf p(x)]$
whose rows are indexed by $\mf t,$ 
\[
\omega_{\mf t}(x) = \det \left[ D^{\ell} \mathbf p_{\mf t(k)}\right]_{\ell, k=0}^{L-1}.
\]

We define
\[
\mathrm{Wr}(\mathbf p_{\mf t}; x) = \det \left[ D^{\ell} \mathbf p_{\mf t(k)}(x) \right]_{\ell, k=0}^{L-1}
\]
to be the (renormalized) {\em Wronskian} of $\mathbf p_{\mf t} := (\mathbf
p_{\mf t(0)}, \ldots, \mathbf p_{\mf t(L-1)}).$ This differs from the usual
Wronskian by a factor of $\prod_{\ell=0}^{L-1} \ell!.$ Thus,
\[
\omega(x) = \sum_{\mf t \in {[N) \choose L}}  \mathrm{Wr}(\mathbf p_{\mf t}; x) \mathbf e_{\mf t}.
\]

To give an explicit example, one which will prove useful in the circular
ensembles, let $\mathbf m(x) = (1, x, \ldots, x^{N-1})$ be the complete family
of {\em monomials}. The Wronskians of collections of monomials is known, 
\begin{lemma}
	Given $\mf t \in {[N) \choose L}$ define 
	\[
	\Sigma \mf t = \mf t(1) + \cdots + \mf t(L), \qq{and} \wt \Delta \mf t = \prod_{j < k}^L \frac{t(k) - t(j)}{k - j}.
	\]
	Then,
	\[
	\mathrm{Wr}(\mathbf m_{\mf t}(x)) = \wt \Delta \mf t \, x^{\Sigma \mf t}.
	\]
\end{lemma}

\section{Results}
	\label{sec:results}
From here forward $L$ is an even positive integer, $\beta = L^2.$ In which case, 
\begin{align*}
	f(\mathbf x) &= \frac{1}{M! Z} \ast \left( \omega(x_1) \wedge \omega(x_2) \wedge \cdots \wedge \omega(x_M) \right) \cdot \prod_{j=1}^M u(x_j).
\end{align*}
Setting $\wt \omega(x) = u(x) \omega(x)$, the joint density of particles and
partition function are given by
\[
f(\mathbf x) = \frac{1}{M! Z} \ast \wt \omega(x_1) \wedge \cdots \wedge \wt \omega(x_M)  \qq{and} Z = \frac{1}{M!} \int_{W^M} \ast \wt \omega(x_1) \wedge \cdots \wedge \wt \omega(x_M) \, d\mu^M(\mathbf x).
\]
The $m$th correlation function is given by
\[
R_m(\mathbf y) = \frac{1}{Z (M-m)!}\int_{W^{M-m}} \ast  \wt \omega(y_1) \wedge \cdots \wedge \wt \omega(y_m) \wedge \wt \omega(x_1) \wedge \cdots \wedge \wt \omega(x_{M-m}) \, d\mu^{M-m}(\mathbf x),
\]
and we define the $L$-vector $\int \wt \omega d\mu \in \Lambda^L V$ by 
\[
\int \wt \omega \, d\mu := \sum_{\mf t \in {(N] \choose L}} \bigg(\int_W \omega_{\mf t}(x) u(x) \, d\mu(x)\bigg) \mathbf e_{\mf t}.
\]
That is, we extend the integral operator $\int \cdot \, d\mu$ to $L$-vectors by
integrating the coefficients. This is independent of basis. There is a Fubini's
theorem for integrals over (coefficients of) multivectors \cite{MR1985318,
MR73174, MR4035418}. Applied to the partition function, it has
\[
Z = \frac{1}{M!} \ast \left( \int_W \wt \omega(x) \, d\mu(x) \right)^{\wedge M}.
\]
This is a hyperpfaffian,
\[
Z = \mathrm{PF} \left( \int_W \wt \omega(x) \, d\mu(x) \right).
\]
We call 
\[
\gamma = \int_W \wt \omega(x) \, d\mu(x)
\]
the {\em Gram} $L$-vector for the ensemble, and it plays an important role in
the analysis of the correlation functions. 
\begin{equation}
	\label{eq:Z}
	Z = \mathrm{PF}(\gamma) \qq{and} R_m(\mathbf y) = \ast\frac{1}{Z}\left( \wt \omega(y_1) \wedge \cdots \wedge \wt \omega(y_m) \wedge \frac{\gamma^{\wedge(M-m)}}{(M-m)!}  \right).
\end{equation}
The Gram $L$-vector depends on the monic polynomials $\mathbf p(x)$ but the
partition function and the correlation functions do not. Part of the art of
analysis of these ensembles will be identifying the monic polynomials that
maximally simplify $\gamma$. Our main, general result is as follows, though this
simplifies considerably in the circular ensembles.
\begin{thm}
	\label{thm:main}
	Let $M' = M - m$ and $N' = LM'.$ 
	\[
	R_m(\mathbf y) = \frac{1}{Z} \prod_{j<k}^m (y_k - y_j)^\beta \cdot \prod_{n=1}^m u(y_n) \cdot \PF \gamma_{\mathbf y},
	\]
	where $\gamma_{\mathbf y}$ is the $L$-vector over $V' =
	\mathrm{span}_{\C}\{\mathbf e_0, \mathbf e_1, \ldots, \mathbf e_{N'-1}\}$
	given by 
	\[
	\gamma_{\mathbf y} = \sum_{\mf u \in {[N') \choose L}} \int_W \bigg[\prod_{j=1}^m (x - y_j)^{\beta} \bigg] \mathrm{Wr} \big(\mathbf p_{\mf u}(x)\big) u(x) \, d\mu(x) \, \mathbf e_{\mf u}.
	\]
\end{thm}

\subsection{Reduction to Pfaffian Point Processes when $\beta=4$}
\label{subsec:pfaffian_point}

When $\beta=4$, we have $L=2$, and hence
\[
\omega(x)
=
\sum_{\mf t\in {[2M)\choose 2}}
\mathrm{Wr}(\mathbf p_{\mf t};x)\,\mathbf e_{\mf t}
=
\sum_{0\leq j<k<2M}
\mathrm{Wr}(p_j,p_k;x)\,
\mathbf e_j\wedge\mathbf e_k.
\]
The $2\times2$ Wronskian is
\[
\mathrm{Wr}(p_j,p_k;x)
=
p_j'(x)p_k(x)-p_j(x)p_k'(x),
\]
so that
\[
\omega(x)
=
\sum_{0\leq j<k<2M}
\left[
p_j'(x)p_k(x)-p_j(x)p_k'(x)
\right]
\mathbf e_j\wedge\mathbf e_k.
\]
The Gram form is therefore
\[
\gamma
=
\sum_{0\leq j<k<2M}
\int
\left[
p_j'(x)p_k(x)-p_j(x)p_k'(x)
\right]
u(x)\,d\mu(x)\,
\mathbf e_j\wedge\mathbf e_k.
\]

Define the skew inner product
\[
\la f,g\ra_u
=
\int
\left[
f'(x)g(x)-f(x)g'(x)
\right]
u(x)\,d\mu(x).
\]
Then
\[
\gamma
=
\sum_{0\leq j<k<2M}
\la p_j,p_k\ra_u\,
\mathbf e_j\wedge\mathbf e_k,
\qquad
Z=\mathrm{PF}(\gamma).
\]
The hyperpfaffian of the $2$-form $\gamma$ is the ordinary
Pfaffian of its associated $2M\times2M$ antisymmetric Gram matrix
\[
\mathbf G
=
\left[
\la p_j,p_k\ra_u
\right]_{j,k=0}^{2M-1}.
\]
Consequently,
\[
Z=\Pf(\mathbf G),
\]
which is the familiar Pfaffian formula for the partition function of
the $\beta=4$ ensemble. This connection motivates the terminology
\emph{Gram form} for $\gamma$.

To specialize Theorem~\ref{thm:main}, define, for
$\mathbf y=(y_1,\ldots,y_m)\in W^m$,
\[
s_{\mathbf y}(x)
=
\prod_{n=1}^m(x-y_n)^2.
\]

\begin{cor}
When $L=2$,
\[
R_m(\mathbf y)
=
\frac{1}{Z}
\prod_{1\leq j<k\leq m}|y_k-y_j|^4
\prod_{n=1}^m u(y_n)
\Pf(\mathbf G_{\mathbf y}),
\]
where $\mathbf G_{\mathbf y}$ is the
$2(M-m)\times2(M-m)$ antisymmetric matrix
\[
\mathbf G_{\mathbf y}
=
\left[
\la s_{\mathbf y}p_j,s_{\mathbf y}p_k\ra_u
\right]_{j,k=0}^{2(M-m)-1}.
\]
\end{cor}

We next relate this modified Gram Pfaffian to the classical
Pfaffian point-process kernel. Let
\[
V_M
=
\operatorname{span}\{p_0,\ldots,p_{2M-1}\}.
\]
Since $p_j$ is monic of degree $j$, this is the space of polynomials of
degree at most $2M-1$. Introduce the confluent evaluation map
\[
J_{\mathbf y}:V_M\longrightarrow\R^{2m},
\qquad
J_{\mathbf y}f
=
\bigl(
f(y_1),f'(y_1),
\ldots,
f(y_m),f'(y_m)
\bigr).
\]
A polynomial belongs to $\ker J_{\mathbf y}$ precisely when it
vanishes to order at least two at each of the points
$y_1,\ldots,y_m$. Consequently,
\[
\ker J_{\mathbf y}
=
\left\{
s_{\mathbf y}q:
q\in\R[x],\
\deg q\leq 2(M-m)-1
\right\}.
\]
Thus
\[
s_{\mathbf y}p_0,\ldots,
s_{\mathbf y}p_{2(M-m)-1}
\]
form a basis of $\ker J_{\mathbf y}$, and
$\mathbf G_{\mathbf y}$ is the Gram matrix of the restriction of the
skew inner product to this constrained polynomial subspace.

To identify the complementary Pfaffian, consider the ordered basis
\[
s_{\mathbf y}p_0,\ldots,
s_{\mathbf y}p_{2(M-m)-1},
p_0,\ldots,p_{2m-1}
\]
of $V_M$. Since $s_{\mathbf y}p_j$ is monic of degree $2m+j$, the
change-of-basis matrix from
$p_0,\ldots,p_{2M-1}$ to this basis has determinant one. On the final
$2m$ basis elements, the map $J_{\mathbf y}$ is represented by the
confluent evaluation matrix
\[
\mathbf H_{\mathbf y}
=
\left[
J_{\mathbf y}p_0
\ \cdots\
J_{\mathbf y}p_{2m-1}
\right].
\]
Because the polynomials $p_j$ are monic,
$\mathbf H_{\mathbf y}$ has the usual confluent Vandermonde
determinant
\[
\det\mathbf H_{\mathbf y}
=
\prod_{1\leq i<j\leq m}(y_j-y_i)^4.
\]

Since $Z=\Pf(\mathbf G)\neq0$, the full Gram matrix $\mathbf G$ is
invertible. Applying the Pfaffian analogue of Jacobi's complementary
minor identity \cite{rains-2000,borodin-2008} in the preceding adapted
basis, and then transporting the complementary inverse block to the
jet coordinates by $\mathbf H_{\mathbf y}$, gives
\[
\prod_{1\leq i<j\leq m}(y_j-y_i)^4
\frac{\Pf(\mathbf G_{\mathbf y})}{\Pf(\mathbf G)}
=
\Pf\left(
-J_{\mathbf y}\mathbf G^{-1}J_{\mathbf y}^{\mathsf T}
\right).
\]

To identify the resulting matrix with the usual Pfaffian kernel, set
\[
\mathbf p(x)
=
\begin{pmatrix}
p_0(x)\\
\vdots\\
p_{2M-1}(x)
\end{pmatrix}
\]
and define the $2\times2M$ jet matrix
\[
\mathcal J(x)
=
\begin{pmatrix}
\mathbf p(x)^{\mathsf T}\\
\mathbf p'(x)^{\mathsf T}
\end{pmatrix}.
\]
Then
\[
J_{\mathbf y}\mathbf G^{-1}J_{\mathbf y}^{\mathsf T}
=
\left[
\mathcal J(y_i)\mathbf G^{-1}
\mathcal J(y_j)^{\mathsf T}
\right]_{i,j=1}^m.
\]
It follows that
\[
R_m(\mathbf y)
=
\Pf\left[
\mathbb K_M(y_i,y_j)
\right]_{i,j=1}^m,
\]
where
\[
\mathbb K_M(x,y)
=
-\sqrt{u(x)u(y)}\,
\mathcal J(x)
\mathbf G^{-1}
\mathcal J(y)^{\mathsf T}.
\]
The factors $\sqrt{u(y_i)}$ multiply both jet coordinates at $y_i$,
and hence contribute the factor $u(y_i)$ to the Pfaffian.

Equivalently, define the scalar prekernel
\[
\kappa_M(x,y)
=
-\mathbf p(x)^{\mathsf T}
\mathbf G^{-1}\mathbf p(y).
\]
Then
\[
\mathbb K_M(x,y)
=
\sqrt{u(x)u(y)}
\begin{pmatrix}
\kappa_M(x,y)
&
\partial_y\kappa_M(x,y)
\\[2mm]
\partial_x\kappa_M(x,y)
&
\partial_x\partial_y\kappa_M(x,y)
\end{pmatrix}.
\]

If the polynomials are chosen to be skew-orthogonal, with
\[
\la p_{2r},p_{2s+1}\ra_u
=
t_r\delta_{rs},
\qquad
\la p_{2r},p_{2s}\ra_u
=
\la p_{2r+1},p_{2s+1}\ra_u
=
0,
\]
then $\mathbf G$ becomes block diagonal and
\[
\kappa_M(x,y)
=
\sum_{r=0}^{M-1}
\frac{
p_{2r}(x)p_{2r+1}(y)
-
p_{2r+1}(x)p_{2r}(y)
}{t_r}.
\]
This is the familiar finite sum of skew-orthogonal polynomial pairs.
Consequently, with the present value-derivative ordering of the jet
coordinates, $\mathbb K_M$ is the usual finite-$M$ matrix kernel for
the classical $\beta=4$ Pfaffian point process. Other common
conventions for the $\beta=4$ kernel are obtained by permuting the two
jet coordinates at each point.

The two representations emphasize complementary aspects of the same
correlation function. The specialization of
Theorem~\ref{thm:main} gives the Pfaffian of the skew form restricted
to the space of polynomials vanishing doubly at the marked points,
whereas the classical kernel formula gives the Pfaffian of the inverse
skew form evaluated on the complementary value and derivative data.

The constrained-subspace interpretation persists for general even
$L$. Indeed, if
\[
V_M^{(L)}
=
\operatorname{span}\{p_0,\ldots,p_{LM-1}\}
\]
and
\[
J_{\mathbf y}^{(L)}f
=
\left(
f^{(r)}(y_i)
\right)_{
\substack{
1\leq i\leq m\\
0\leq r<L
}},
\]
then, with
\[
s_{\mathbf y}(x)
=
\prod_{i=1}^m(x-y_i)^L,
\]
one has
\[
\ker J_{\mathbf y}^{(L)}
=
\left\{
s_{\mathbf y}q:
q\in\R[x],\
\deg q\leq L(M-m)-1
\right\}.
\]
Thus Theorem~\ref{thm:main} may again be interpreted as the
hyperpfaffian of the Gram $L$-form restricted to a polynomial
subspace satisfying confluent vanishing conditions.

What is special to $L=2$ is the passage from this restricted Gram
form to a kernel. In that case, the coefficients of $\gamma$ correspond to entries of an
antisymmetric matrix; moreover, since the Pfaffian of $\gamma$ is non-vanishing, the associated matrix has an inverse, and Jacobi's complementary minor identity
expresses the Pfaffian of a restriction of the associated matrix in terms of the corresponding
minor of the inverse matrix. However, for $L>2$, the Gram object is an
alternating $L$-tensor rather than a matrix. It has no canonical
matrix inverse, and there is no presently known general hyperpfaffian
analogue of Jacobi's identity that converts the hyperpfaffian of the
restricted Gram form into a hyperpfaffian of complementary inverse
data. Consequently, Theorem~\ref{thm:main} gives an exact
hyperpfaffian formula for the correlation functions, but it does not
by itself produce a classical finite-rank kernel representation.
Obtaining such a representation would require an additional theory of
hyperpfaffian adjugates.

\subsection{Circular $\beta$-ensembles}
By symmetry, in the circular case, we know the monomials must be the best choice
to simplify $\gamma$. If $\mu$ is Haar probability measure on $\T$, 
\[
\int_{\T} \mathrm{Wr}(\mathbf m_{\mf t}(x))  u(x) \, d\mu(x) = \wt \Delta \mf t \int_{\T} x^{\Sigma \mf t - L(N-1)/2} \, d\mu(x) = \piecewise{\wt \Delta \mf t & \mbox{if } \Sigma \mf t = L(N-1)/2; \\ 0 & \mbox{otherwise.}}
\]
Note that if we were to choose $L$ integers from $[N)$ uniformly and
independently, then their expected sum is $\overline \Sigma := L(N-1)/2$, thus
we may represent the Gram $L$-vector by
\[
\gamma =  \sum_{\Sigma \mf t = \overline \Sigma} \wt \Delta \mf t \, \mathbf e_{\mf t} \quad \in \quad \Lambda^L V,
\]
where the sum is over all $\mf t \in {[N) \choose L}$ such that $\Sigma \mf t =
\overline \Sigma$. We note that $\wt \Delta \mf t$ is an integer, and $\gamma$
is fairly sparse in the sense that most of the coefficients are equal to zero.

\begin{thm}\label{thm:circ-thm} Let $M' = M - m$ and $N' = LM'.$ Given $\mf u
	\in {[N') \choose L}$ let $\delta \mf u = \Sigma \mf u - L(N'-1)/2.$ Then,
	\[
	R_m(\mathbf y) = M!{\frac{\beta M}2 \choose \frac{\beta}2, \cdots, \frac{\beta}2 }^{-1} \prod_{j<k}^m |y_k - y_j|^\beta \cdot \prod_{n=1}^m y_n^{-\beta(M-m)/2} \cdot \PF \gamma_{\mathbf y},
	\]
	where $\gamma_{\mathbf y}$ is the $L$-vector over $V' =
	\mathrm{span}_{\C}\{\mathbf e_0, \mathbf e_1, \ldots, \mathbf e_{N'-1}\}$
	given by 
	\[
	\gamma_{\mathbf y} = \sum_{|\delta \mf u| \leq \beta m/2} \bigg[ x^{-\beta m/2} \prod_{j=1}^m (x - y_j)^{\beta} \bigg]_{(\delta \mf u)} \wt \Delta \mf u \, \mathbf e_{\mf u},
	\]
	and the coefficient of $x^j$ of the Laurent polynomial $\ell(x)$ is denoted
	by $[\ell(x)]_{(j)}$. The sum is over all $\mf u \in {[N') \choose L}$ with
	$|\delta \mf u | \leq \beta m/2$. 
\end{thm}

With a few definitions we may rephrase this theorem more succinctly. Given $j
\in \Z,$ define the $L$-vector
\[
\epsilon^j = \sum_{\delta \mf u = j} \wt \Delta \mf u \, \mathbf e_{\mf u} \quad \in \quad \Lambda^L V'
\]
where the sum is over all $\mf u \in {[N') \choose L}$ such that $\delta \mf u =
j.$ Note the superscript is simply a convenience and not representing that this
is a power in the exterior algebra. Define $f_{\mathbf y}(x) = x^{-m/2}
\prod_{n} (x - y_n)$ and define $b_j(\mathbf y), |j| \leq \beta m/2$  to be the
coefficients of $f^\beta_{\mathbf y}(x),$
\[
f^\beta_{\mathbf y}(x) = x^{-\beta m/2} s^L_{\mathbf y}(x) = \sum_{j=-\beta m/2}^{\beta m/2} b_j(\mathbf y) x^j.
\]
By rotational symmetry we may assume without loss of generality that $b_{-\beta
m/2} = b_{\beta m/2} = 1,$ and because $f^\beta_{\mathbf y}$ is a conjugate
reciprocal Laurent polynomial, $b_{-j} = \overline b_j$ for all $j \leq \beta
m/2.$ In particular, $b_0$ is real. At any rate, we may superficially define
\[
f^\beta_{\mathbf y}(\epsilon) := \sum_{j=-\beta m/2}^{\beta m/2} b_j(\mathbf y) \epsilon^j \quad \in \quad \Lambda^L V'.
\]
In spite of appearances, the right hand side is not a polynomial, but rather a
linear combination of the $\epsilon^j$ above. This notation really is cheating,
but it is very succinct and shows how the $\beta$ power of conjugate reciprocal
Laurent polynomials parametrize the $L$-vectors whose hyperpfaffians yield the
correlation functions. That is,
\begin{cor}
	\label{cor:moduli}
	\[
	R_m(\mathbf y) = M! {\frac{\beta M}2 \choose \frac{\beta}2, \cdots, \frac{\beta}2 }^{-1}  \prod_{j<k}^m |y_k - y_j|^\beta \cdot \prod_{n=1}^m y_n^{-\beta(M-m)/2} \cdot \PF f^\beta_{\mathbf y}(\epsilon).
	\]
\end{cor}

\subsubsection{Pair Correlation}\label{sec:pair-cor}

We turn to the second correlation function, and in particular we give it in a
form in sympathy with the code used to compute the examples presented here. 

Here $m=2,$ $N'=L(M-2)$ and $V' = \mathrm{span}_{\C}\{\mathbf e_0, \ldots,
\mathbf e_{N'-1}\}.$ By the rotational symmetry of the circle, we may assume
that $y_1 = e^{i \theta}$ and $y_2 = \overline{y_1} = e^{-i \theta}$ for some
$\theta \in [0, \pi]$. Then, we may write $\wt \eta_\theta := \wt \eta_{(y_1,
y_2)}$ and
\[
\int_{\T} \wt \eta_{\theta}(x) d\mu(x) = \sum_{|\delta \mf u| \leq \beta} \bigg[ \bigg(x + \frac1x - 2 \cos \theta\bigg)^{\beta} \bigg]_{(\delta \mf u)} \wt \Delta \mf u \, \mathbf e_{\mf u},
\]
where the sum is over $\mf u \in {[N') \choose L}$ with $|\delta \mf u| \leq
\beta.$ An easy calculation reveals,
\[
b_j(\theta) := \bigg[ \bigg(x + \frac1x - 2 \cos \theta\bigg)^{\beta} \bigg]_j = \sum_{\ell=|j|}^{\beta} {\beta \choose \ell} {\ell \choose \frac{\ell+|j|}{2}} (-2 \cos \theta)^{\beta - \ell}.
\]
Let us define $E = \{\mathbf j = (j_1 \leq j_2 \leq \cdots \leq j_{M-2}) : |j_n|
\leq \beta, \Sigma \mathbf j = 0\}.$ Given $|j| \leq \beta$ we define the
multiplicity of $j$ in $\mathbf j$ by $N_{\{j\}}(\mathbf j),$ and the
multinomial coefficient 
\[
\mathrm{mult}(\mathbf j) = (M-2)! \prod_{|j| \leq \beta} \frac{1}{N_{\{j\}}(\mathbf j)!}.
\]
Then,
\[
R_2(\theta) = \frac{M!}{(M-2)!} {\frac{\beta M}2 \choose \frac{\beta}2, \cdots, \frac{\beta}2 }^{-1}  \frac{(2 \sin \theta)^{\beta}}{2 \pi} \sum_{\mathbf j \in E} \mathrm{mult}(\mathbf j) \prod_{n=1}^{M-2} b_{j_n}(\theta) \cdot \ast \bigwedge_{n=1}^{M-2} \epsilon^{j_n}.
\] 
Our algorithm for computing $R_2$ from this is now clear. The only
computationally complex components are the calculation of $E$ and the exterior
product $\bigwedge_{n=1}^{M-2} \epsilon^{j_n}.$

\subsection{Hyperpfaffian Evaluations}

We conclude with a detour from correlations to talk about hyperpfaffian
evaluations, which add to the existing class of pfaffian and hyperpfaffian
formulas demonstrated by Ishikawa and Zheng in \cite{ishikawa2022}. In general,
hyperpfaffian evaluations are hard. Without some special structure, a typical
element in $\Lambda^L V$ will have ${N \choose L}$ non-zero coefficients, and
the hyperpfaffian will be an $M$th power of this. Of course, we expect many
terms to annihilate when taking powers, but it is nonetheless computationally
expensive as a sort is necessary to determine signs of terms which do not
annihilate. All this is to say, choosing monic families of polynomials $\mathbf
p(x)$ and/or basis elements $\mathbf e_1, \ldots, \mathbf e_N$ for which
$\gamma$ and $\gamma_{\mathbf y}$ have a maximal number of non-zero coefficients
is useful for calculations, and we expect will be useful for proving further
theorems about the correlations in specific $\beta$-ensembles.

There are a couple `easy' hyperpfaffian evaluations. For instance, the
multivector
\[
\xi = c_0 \cdot \mathbf e_0 \wedge \cdots \wedge \mathbf e_{L-1} + c_1 \cdot \mathbf e_{L} \wedge \cdots \wedge \mathbf e_{2L-1} + \cdots + c_{M-1} \cdot \mathbf e_{L(M-1)} \wedge \cdots \wedge \mathbf e_{LM-1}
\]
is an example of a {\em diagonal} form, and up to permutations of the basis
elements diagonal multivectors are the simplest which can have a non-zero
hyperpfaffian,
\[
\PF \xi = c_0 c_1 \cdots c_{M-1}.
\]
This hyperpfaffian is the analog of the first Pfaffian evaluation in
\ref{eq:pfaffs}. If we could always find a monic family of polynomials such that
the Gram form $\gamma$ was diagonal, we could perform similar maneuvers to those
which produce the matrix kernel in the $\beta = 4$ Pfaffian point process to
produce an $L$-vector kernel which played the same role in the $\beta = L^2$
case. Unfortunately, this seems to be too much to ask for.

However, while formulas for the correlations are lacking in $\beta$-ensembles,
in many cases the formulas for the partition function are known. Many of these
follow from evaluations of the {\em Selberg integral} (and its kin):
\[
\int_{[0,1]^M} \bigg\{\prod_{n=1}^M x_n^{a-1} (1 - x_n)^{b-1} \bigg\} \prod_{j < k} |x_k - x_j|^{2c} dx_1 \cdots dx_M = \prod_{n=0}^{M-1} \frac{\Gamma(a + nc) \Gamma(b + n c) \Gamma((n+1)c + 1)}{\Gamma(a + b + (M + n - 1)c) \Gamma(1 + c)}.
\]
See \cite{MR2434345} for a more complete history of the Selberg integral and its
variations. The relevance of this is immediately clear; when $\beta = 2c$ we
arrive at an evaluation of the partition function for the $\beta$-ensembles with
Jacobi weight $u(x) = \bs 1_{[0,1]}(x) x^{a-1} (1-x)^{b-1}.$ And, because we
know the partition function is hyperpfaffian, we get an explicit hyperpfaffian
evaluation using the moments of the beta distribution.
\begin{prop}
	Let $B(a,b)$ be the Beta function. Then, 
	\[
	\mathrm{PF}\left(\sum_{\mf t \in {[N) \choose L}} \frac{B(a + \Sigma \mf t, b)}{B(a,b)}\, \wt \Delta \mf t \, \mathbf e_{\mf t}\right) = \frac{1}{M!} \prod_{n=0}^{M-1} \frac{\Gamma(a + n\beta/2) \Gamma(b + n \beta/2) \Gamma((n+1)\beta/2 + 1)}{\Gamma(a + b + (M + n - 1)\beta/2) \Gamma(1 + \beta/2)}.
	\]
\end{prop}

A similar formula for the partition functions for $\beta$-ensembles with Hermite
(Gaussian) weight $u(x) = e^{-x^2/2}$ is known as the {\em Mehta integral}.  By
modifying the family of polynomials we get infinitely many different
hyperpfaffian evaluations using the Mehta integral. For instance, if we use the
monomials, we then get the following hyperpfaffian formula using the moments of
normal random variables.
\begin{prop}
	Let $E$ be the subset of ${[N) \choose L}$ such that $\Sigma \mf t$ is even.
	Then,
	\[
	\mathrm{PF}\left(\sum_{\mf t \in E} (\Sigma \mf t)!! \, \wt \Delta \mf t \, \mathbf e_{\mf t}\right) = \frac{1}{M!} \prod_{n=1}^M \frac{(\beta n/2)!}{(\beta/2)!},
	\]
	where $(2j)!! = 2j \cdot (2j - 2) \cdots 4 \cdot 2.$ The sum can be taken over all $\mf t
	\in {[N) \choose L}$ by replacing the double factorial with the appropriate
	moment of a standard normal random variable (all odd moments are all zero).
\end{prop}
Without prescient knowledge as to which polynomials might maximally simplify the
Gram $L$-vector, our most natural starting family is the monic Hermite
polynomials $\mathbf h = (h_0, \ldots, h_{N-1}).$ The evaluation of the Mehta
integral, then produces the following:
\begin{prop}
	\label{thm:hermite-eval}
	Let $\mathbf h = (h_0, \ldots, h_{N-1})$ be given by the monic Hermite
	polynomials orthogonal to the weight $u(x) = e^{-x^2/2}/\sqrt{2 \pi}$ then, 
	\[
	\mathrm{PF}\left(\sum_{\mf t \in {[N) \choose L}}\frac{\mathbf e_{\mf t}}{\sqrt{2 \pi}}\int_{\R} e^{-x^2/2} \mathrm{Wr}(\mathbf h_{\mf t}(x)) \, dx \right) = \frac{1}{M!} \prod_{n=1}^M \frac{(\beta n/2)!}{(\beta/2)!}.
	\]
\end{prop}
When $\beta = 4$ (the symplectic ensembles), the Gram $2$-vector is identified
with the antisymmetric Gram matrix for the ensemble, and the families of {\em
skew-orthogonal} polynomials which maximally simplify the Gram matrix for the
classic weights are known \cite{MR1762659}. The polynomials we seek are thus
generalizations of orthogonal and skew-orthogonal polynomials, and are related to Wronskians of the related orthogonal polynomials. Historically, the study of Wronskians of the classic orthogonal polynomials revolved around the
determination that (in many instances, at least) $2 \times 2$ Wronskians of
orthogonal polynomials have no real zeros \cite{MR142972}. Developments in
technology have caused the study of Wronskians of orthogonal polynomials to
explode over the last couple of decades following (among other things) the
experimental observation of the zeros of higher dimensional Wronskians in the
complex plane \cite{MR4075920, MR4197157, MR2998116, MR3990770}. It is worth
doing some experimentation on one's own, but the extreme rigidity/patterns
formed by the zeros of Wronskians of orthogonal polynomials inspired a number of
interesting observations, conjectures and theorems. It is beyond our scope to
survey the recent literature. However, the progress in our understanding of
these Wronskians is unquestionably relevant to our understanding and eventual
closed form calculation of $\gamma$ and $\gamma_{\mathbf y}$ (and their
hyperpfaffians) for $\beta = L^2$ ensembles with classical weights.

Likewise the integral of the $\beta$ power of the absolute Vandermonde on the
torus has a known evaluation conjectured by Dyson \cite{MR0143556} and proved by
Gunson \cite{gunson:752}, which produces the following hyperpfaffian evaluation.
\begin{prop}
	\label{thm:eval}
	\[
	\mathrm{PF}\left( \sum_{\mf t : \Sigma \mf t = \overline \Sigma} \wt \Delta \mf t \, \mathbf e_{\mf t}\right) = \frac{1}{M!} {\frac{\beta M}2 \choose \frac{\beta}2, \cdots, \frac{\beta}2 },
	\]
	where the sum is over all $\mathbf t \in {[N) \choose L}$ such that $\Sigma
	\mf t = \overline \Sigma$ (or equivalently $\delta \mf t = 0$).
\end{prop}
The first correlation function in the circular case is $R_1 : \T \rightarrow [0,
\infty)$ such that for any Borel subset $B$ of $\T$, $\int_B R_1 \, d\mu =
\mathbf E[N_B] = M \mu(B).$ The final equality is simply rotational invariance,
as we expect $B$ to have $N_B$ proportional to its Haar measure. It follows that
$R_1 = M$ ($\mu$-a.e.). (Note that this line of argumentation will not work for
the non-circular weights.) This gives us another way to compute the partition
function, by Corollary~\ref{cor:moduli}, and another hyperpfaffian evaluation.
\[
R_1(y) = M = \frac{1}{Z} y^{-\beta (M-1)/2} \cdot \mathrm{PF} \left(y^{-\beta/2} (y - \epsilon)^{\beta}\right).
\]
\begin{prop}
	\[
	\mathrm{PF}\left(\sum_{|\delta \mf u| \leq \beta/2} {\beta \choose \delta \mf u + \beta/2} (-y)^{\delta \mf u}\wt \Delta \mf u \, \mathbf e_{\mf u}\right) = \frac{1}{(M-1)!} {\frac{\beta M}2 \choose \frac{\beta}2, \cdots, \frac{\beta}2 } y^{\beta (M-1)/2}, 
	\]
	where the sum is over all $\mf u \in {[L(M-1)) \choose L}$ such that
	$|\delta \mf u| \leq \beta/2.$
\end{prop}

We give one final hyperpfaffian evaluation, which follows from the fact that
$\omega(x) \wedge \omega(x) = 0.$ 
\begin{prop}
	For any $x \in \C$, 
	\[
	\mathrm{PF}\left(\sum_{\mf t \in {[N) \choose L}} x^{\Sigma \mf t} \wt \Delta \mf t \, \mathbf e_{\mf t}\right) = 0.
	\]
\end{prop}

\subsection{Future Directions}
\label{subsec:future}

Here, we show that for $\beta$-ensembles with $\beta = L^2$ an even integer, the
partition function and all correlation functions admit an exact hyperpfaffian
formulation, extending the classical Pfaffian structure of the $\beta = 4$ case
to this broader family. The central object is the Gram $L$-vector $\gamma$,
built from the weighted basis $L$-vector $\wt\omega$, whose hyperpfaffian gives
the partition function $Z$ and whose wedge powers give each correlation function
$R_m$ in closed form (Theorem~\ref{thm:main}). We specialized this general
result to the circular $\beta$-ensembles, where the relevant integrals can be
evaluated explicitly, yielding hyperpfaffian formulas for the one and two point
functions $R_1$ and $R_2$.

A natural next direction is to investigate the asymptotic behavior of the
hyperpfaffian correlation functions as $M \to \infty$. In the classical $\beta =
1, 2, 4$ ensembles, such limits are well understood and yield universal local
statistics: sine kernel scaling in the bulk, Airy-type scaling at a soft edge,
and Bessel-type scaling at a hard edge (see \cite{Kuijlaars2015} for a summary
of universality scaling for classical ensembles; detailed treatments appear in
many texts, including \cite{MR2129906, IntRMT, Forrester2010}). It would be
useful to assess whether analogous universal limits emerge from the
hyperpfaffian formula of Theorem~\ref{thm:main}, and whether the limiting
objects depend on $L$ in an essential way or stabilize as $L$ grows. 

The current paper sets out a necessary first step in this direction by providing
explicit hyperpfaffian formulas for certain correlation functions. In the
classical cases, universal scaling limits can be obtained using the determinant
/ pfaffian generated from a bivariate kernel function, which arises from the
determinantal / pfaffian structure of the two-point correlation function, and
further analysis of the kernel makes use of the theory of orthogonal
polynomials, as well as the Christoffel-Darboux identity, as presented in
\cite{MR1657844,MR1675356}. At present, however, the Gram form
$\gamma$ for $L > 2$ does not readily admit a kernel presentation; the method
used in Subsection~\ref{subsec:pfaffian_point} to produce a kernel in the $\beta = 4$ case ($L = 2)$ utilizes a
change-of-basis transformation to express the Gram form using
non-redundant indices (or equivalently, expressing the associated antisymmetric
matrix in symplectic block-diagonal form). However, an analogous diagonalization
does not exist for generic $L$-forms when $L > 2$. Thus, the absence of a clear
kernel-like object as is used in the proofs of asymptotic results for the $\beta
= 1, 2, 4$ ensembles prevents immediate generalization to the $\beta = L^2$
case. Nevertheless, identifying a suitable kernel analogue appears to be a
productive avenue for further exploration.

As a result of the latent algebraic structure, the current manuscript discusses
correlation functions only for $\beta$ parameters restricted to the sparse set
of even square integers. However, $\beta$-ensemble models for random matrices
and for log gases can be defined for arbitrary non-negative $\beta$. It is
likely that similar results for correlation functions can be extended to $\beta$
following other special integer sequences, just as results for hyperpfaffian
partition functions were extended from even to odd square $\beta$ in
\cite{springerlink:10.1007/s00605-011-0371-8, MR4035418, MR4480836}, or to
integer $\beta$ adjacent to squares in
\cite{springerlink:10.1007/s00605-011-0371-8}. Moreover, while arbitrary
$\beta$-ensembles may not have correlation functions that admit explicit
hyperpfaffian formulations, the level statistics for $\beta$-ensembles on a
continuous range might be interpolated using discrete integer indices, as is
performed numerically for the level statistics of many-body localizations in
\cite{PhysRevLett.122.180601}.

\section{Proofs}
	\label{sec:proofs}
\subsection{The Proof of Theorem~\ref{thm:main}}

First a lemma
\begin{lemma}
	\label{lemma:wrons}
	Let $f_1, \ldots, f_L$ and $g$ be $(L-1)$-differentiable, then
	\[
	\mathrm{Wr}(g f_1, \ldots, g f_L) = g^L \mathrm{Wr}(f_1, \ldots, f_L).
	\]
\end{lemma}
\begin{proof}
	For $f,g$ which are $(\ell-1)$-differentiable, the product rule gives
	\[
	D^{\ell} (f g) = \sum_{j=0}^{\ell}  D^j f \cdot D^{\ell-j} g.
	\]
	We thus see the matrix defining $\mathrm{Wr}(g f_1, \ldots, g f_L)$ is equal
	to that for $\mathrm{Wr}(f_1, \ldots, f_L)$ times a triangular $L \times L$
	matrix with diagonal entries equal to $g$. 
\end{proof}

Recall,
\[
R_m(\mathbf y) = \ast\frac{1}{Z}\left( \wt \omega(y_1) \wedge \cdots \wedge \wt \omega(y_m) \wedge \frac{1}{M'!}\left( \int_W \wt \omega(x) \, d\mu(x) \right)^{\wedge M'} \right).
\]
In order to write this in terms of a hyperpfaffian we introduce another family
of monic polynomials. Set 
\begin{align*}
	\mathbf q_y(x) = (&1, (x - y),  \cdots, (x - y)^{L-1}), \qquad s_y(x) = (x - y)^L, \qq{and} s_{\mathbf y}(x) = \prod_{j=1}^m s_{y_j}(x).
\end{align*}
We may then make a vector of monic polynomials by 
\begin{align*}
	\mathbf q_\mathbf y(x) = \big(&\mathbf q_{y_1}(x), \\
	&s_{y_1}(x) \mathbf q_{y_2}(x), \\ &s_{y_1}(x) s_{y_2}(x) \mathbf q_{y_3}(x), \\ &\vdots \\ &s_{y_1}(x) \cdots s_{y_{m-1}}(x) \mathbf q_{y_m}(x), \\
	&p_0(x) s_{\mathbf y}(x), p_1(x) s_{\mathbf y}(x), \cdots,  p_{N'-1}(x) s_{\mathbf y}(x)
	\big).
\end{align*}
Denote by $U$ the $Lm$-dimensional subspace of $V$ spanned by $\mathbf e_0,
\ldots, \mathbf e_{Lm-1}$. Then, from the vanishing of $s_{\mathbf y}$ on $y_1,
\ldots, y_m$, with these polynomials we have 
\[
\omega(y_1) \wedge \cdots \wedge \omega(y_m) \quad \in \quad \Lambda^{Lm}(U).
\]
That is, 
\[
\omega(y_1) \wedge \cdots \wedge \omega(y_m) = \det(\mathbf U) \, \mathbf e_0 \wedge \cdots \wedge \mathbf e_{Lm-1},
\]
where $\mathbf U$ is the $Lm \times Lm$ matrix formed from the vectors appearing in
$\omega(y_1) \wedge \cdots \wedge \omega(y_m)$. The polynomials $\mathbf
q_{\mathbf y}(x)$ were designed to make this matrix triangular, and an easy
calculation reveals 
\[
\det(\mathbf U) = \prod_{j < k}^m (y_k - y_j)^{\beta}.
\]
Using this, and taking into account $u(y_1) \cdots u(y_m)$ we find
\[
\wt \omega(y_1) \wedge \cdots \wedge \wt \omega(y_m) = \prod_{j<k}^m (y_k - y_j)^{\beta} \cdot \prod_{n=1}^m u(y_n) \cdot \mathbf e_0 \wedge \cdots \wedge \mathbf e_{Lm-1},
\]
and
\[
R_m(\mathbf y) = \ast \frac{1}{Z} \prod_{j<k}^m (y_k - y_j)^{\beta} \cdot \prod_{n=1}^m u(y_n) \cdot \mathbf e_0 \wedge \cdots \wedge \mathbf e_{Lm-1}  \wedge \frac{1}{M'!}\left( \int_W \wt \omega(x) \, d\mu(x) \right)^{\wedge M'}.
\]
But now we see that only the coefficient of $\mathbf e_{Lm} \wedge \cdots \wedge
\mathbf e_{LM-1}$ of 
\[
\frac{1}{M'!}\left( \int_W \wt \omega(x) \, d\mu(x) \right)^{\wedge M'} \in \quad \Lambda^{N'}V, 
\]
will complement the $\mathbf e_0 \wedge \cdots \wedge \mathbf e_{Lm-1}$
appearing from  $\wt \omega(y_1) \wedge \cdots \wedge \wt \omega(y_m)$. Put
another way, if we set $U^\perp$ to be the span of $\mathbf e_{Lm}, \cdots,
\mathbf e_{N-1}$, then we may replace $\int \wt \omega d\mu$ with its image
under the canonical projection $\Lambda^L V \rightarrow \Lambda^L U^{\perp}$
without changing $R_m(\mathbf y)$. Let us write $\mathbf g_0, \ldots, \mathbf
g_{N'-1}$ for $\mathbf e_{Lm}, \ldots, \mathbf e_{N-1}$ so that $\mathbf e_0,
\ldots, \mathbf e_{Lm-1}, \mathbf g_0, \ldots, \mathbf g_{L(M-m)-1}$ is our
original basis for $V$. 

It follows that if we set 
\[
\eta_{\mathbf y}(x) = \sum_{\mf u \in {[N') \choose L}} \mathrm{Wr}(s_{\mathbf y}(x) \mathbf p_{\mf u}(x) ) \mathbf g_{\mf u} \quad \in \quad \Lambda^L U^\perp,
\]
then
\[
R_m(\mathbf y) = \ast \frac{1}{Z} \prod_{j<k}^m (y_k - y_j)^{\beta} \cdot \prod_{n=1}^m u(y_n) \cdot \mathbf e_0 \wedge \cdots \wedge \mathbf e_{Lm-1}  \wedge \frac{1}{M'!}\left( \int_W \eta_{\mathbf y}(x) u(x) \, d\mu(x) \right)^{\wedge M'},
\]
and from the orthogonality of $U$ and $U^\perp$, we conclude 
\[
R_m(\mathbf y) = \frac{1}{Z} \prod_{j<k}^m (y_k - y_j)^{\beta} \cdot \prod_{n=1}^m u(y_n) \cdot \mathrm{PF}\left(\int_W \eta_{\mathbf y}(x) u(x) \, d\mu(x) \right).
\]
We already see a hyperpfaffian formulation for $R_m(\mathbf y)$, but we can
simplify this further using Lemma~\ref{lemma:wrons}
\[
\eta_{\mathbf y}(x) = \sum_{\mf u \in {[N') \choose L}} s_{\mathbf y}(x)^L \mathrm{Wr}(\mathbf p_{\mf u}(x) ) \mathbf g_{\mf u},
\]
and hence integrating this over $W$ we find $\gamma_{\mathbf y}$ appearing in
the hyperpfaffian and the theorem follows.

\subsection{Proof of Theorem~\ref{thm:circ-thm}}

We have already seen the Wronskians of monomials, and so 
\[
\wt{\eta}_{\mathbf y}(x) = u(x) \eta_{\mathbf y}(x) = \sum_{\mf u \in {[N') \choose L}}  x^{\delta \mf u-\beta m/2} s^L_{\mathbf y}(x)  \wt \Delta \mf u \, \mathbf g_{\mf u}.
\]
The monomial $x^{-\beta m/2}$ `centers' the Laurent polynomial
\[
x^{-\beta m/2} \prod_{j=1}^m (x - y_j)^{\beta}
\]
around its `middle' coefficient. That is, the constant coefficient of this
Laurent polynomial is the central coefficient of $s^L_{\mathbf y}(x)$, and which
can therefore be obtained when $\delta \mf u = 0$ by integration around the unit
circle. More generally, integration of 
\[
\bigg\{ x^{-\beta m/2} \prod_{j=1}^m (x - y_j)^{\beta} \bigg\} \cdot x^{\delta \mf u}
\]
around the unit circle will return a nonzero coefficient of $s^L_{\mathbf y}(x)$
only if
\[
-\frac{\beta m}2 \leq \delta \mf u \leq \frac{\beta m}2,
\]
as values of $ \delta \mf u $ outside this range will correspond to Laurent
series with vanishing constant terms. Therefore, we can write 
\[
\int_{\T} \wt \eta_{\mathbf y}(x) d\mu(x) = \sum_{|\delta \mf u| \leq \beta m/2} \bigg[ x^{-\beta m/2} \prod_{j=1}^m (x - y_j)^{\beta} \bigg]_{(\delta \mf u)} \wt \Delta \mf u \, \mathbf g_{\mf u},
\]
where the sum is over all $\mf u \in {[N') \choose L}$ with $|\delta \mf u |
\leq \beta m/2$. This completes the proof.

\newpage

\appendix
\renewcommand{\thesubsection}{\Alph{subsection}}
\section*{Appendix}

The identities in \ref{sec:pair-cor} provide formula for the pair correlation
functions $R_2(\theta)$ in circular $\beta$-ensembles and can be used to find
explicit formula for $R_2(\theta)$ which are polynomial in $\cos \theta$. While
generating these formula is currently only computationally feasible when $M$ is
small, it is likely that some efficiency improvements can be achieved through
revisions to the computational algorithm that make use of parallelization or
recursive structures. In any case, we anticipate that these explicit formula
will be useful in identifying and predicting asymptotically dominant terms as
either $M$ or $\beta$ grow to $\infty$. 

Below are expressions for the pair correlation functions $R_2(\theta)$ when
$\beta = 16$ and $M \in \{4, 5, 6\}$.

\subsubsection*{Case I: $\beta = 16, M = 4$}
\[
R_2(\theta) = \frac{12}{2\pi \cdot 99561092450391000}(2 \sin \theta)^{16} r(2 \cos \theta),
\]
where 
\begin{align*}
	r(y) &= 12870 \cdot y^{32} + 320320 \cdot y^{30} + 22994400 \cdot y^{28} + 268195200 \cdot y^{26} + 5071284400 \cdot y^{24} \\ &+ 23874264960 \cdot y^{22} + 215207952960 \cdot y^{20} + 254763308800 \cdot y^{18} + 2436174140400 \cdot y^{16} \\ &- 2292869779200 \cdot y^{14}  + 12661736447360 \cdot y^{12} - 21738014538240 \cdot y^{10} + 36816650270400 \cdot y^{8} \\ &- 43095224972800 \cdot y^{6}  + 35720422982400 \cdot y^{4} - 18426562452480 \cdot y^{2} + 4465830320120.
\end{align*}

\begin{figure}[htbp] 
	\centering
	\includegraphics[width=4in]{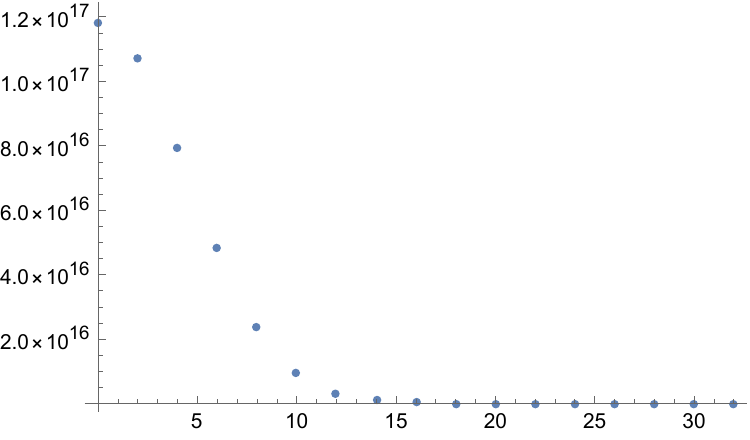}
	\caption{\label{fig:m4fourier} A plot of the non-zero Fourier coefficients
		of $r(2 \cos \theta)$ when $\beta=16$ and $M=4.$}
\end{figure}

It follows that,
\begin{align*}
	r(2 \cos \theta) &= 118075131722187900 + 213766603488921600 \cos (2 \theta )\\ &+158481210768192000 \cos (4 \theta )  +96065488366848000 \cos (6 \theta ) \\ &+47480325016924800 \cos (8 \theta )+19055181216614400 \cos
	(10 \theta ) \\ & +6176104576012800 \cos (12 \theta ) +1604801113344000 \cos (14 \theta ) \\ & +331315058646000 \cos (16 \theta ) +53682163292160 \cos (18 \theta
	)\\ &+6731088698880 \cos (20 \theta )+638896527360 \cos (22 \theta )\\ &+44999094400 \cos (24 \theta )+2230425600 \cos (26 \theta ) \\ & +77975040 \cos (28 \theta )+1464320
	\cos (30 \theta ) \\&+25740 \cos (32 \theta ).
\end{align*}

\newpage

\subsubsection*{Case II: $\beta = 16, M = 5$}
\[
R_2(\theta) = \frac{20}{2\pi \cdot 7656714453153197981835000}(2 \sin \theta)^{16} r(2 \cos \theta),
\]
where 
\begin{align*}
	r(y) &= 9465511770 \cdot y^{48}-26726150880 \cdot y^{46}\\ & +1017078519600 \cdot y^{44} +47238601924800 \cdot y^{42}  \\ & -309243642107400
	\cdot y^{40}+3582382328965440 \cdot y^{38}   \\ & +8316664938822240 \cdot y^{36}   -150541961420822400 \cdot y^{34}  \\ &+1380760795798827300
	\cdot y^{32} -4467019703704272000 \cdot y^{30} \\ & +1140384068909616960 \cdot y^{28}    +72058305100576354560 \cdot y^{26}  \\ & -347152302588287196000
	\cdot y^{24}  +742339143220330656000 \cdot y^{22}   \\ & -116898029205563548800 \cdot y^{20}     -3557317781523015544320
	\cdot y^{18} \\ & +9104214857906943776430 \cdot y^{16}   -8019052421829687295200 \cdot y^{14}  \\ &  -5679307289719178715600
	\cdot y^{12}   +17996708682478726200000 \cdot y^{10}  \\ & -10152703343080717178760 \cdot y^8  -5002967022288185396160
	\cdot y^6 \\ & +4575147589326263320800 \cdot y^4  +1003927326173995766400 \cdot y^2\\ & +17725775603742191700.
\end{align*}

\begin{figure}[htbp] 
	\centering
	\includegraphics[width=4in]{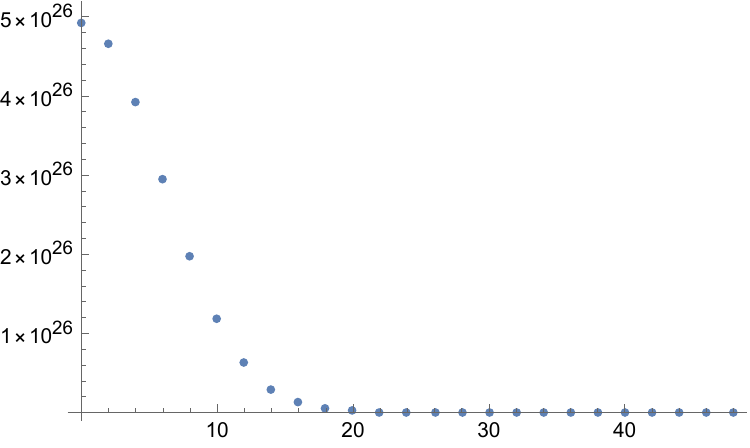}
	\caption{\label{fig:m5fourier} A plot of the non-zero Fourier coefficients
	of $r(2 \cos \theta)$ when $\beta=16$ and $M=5.$}
\end{figure}	

\newpage

It follows that
\begin{align*}
	r(2 \cos \theta) &= 246563699858183708375661000 + 465727370420524755793536000 \cos (2 \theta ) \\ 
	&+392285293376234908519584000 \cos (4 \theta ) +294591250231999404038784000 \cos (6 \theta ) \\  
	&+197120766096092961383976000 \cos (8 \theta )+117431537219232058982016000 \cos (10 \theta ) \\
	&+62216406700671716385235200 \cos (12 \theta )+29275482236971810871116800 \cos (14 \theta) \\
	&+12214266507988416658729800 \cos (16 \theta )+4509469834211802982579200 \cos (18 \theta ) \\
	&+1469813045677907747731200 \cos (20 \theta )+421773093442219018705920 \cos (22 \theta ) \\
	&+106203105750701891954880 \cos (24 \theta )+23378305018893834746880 \cos (26 \theta ) \\
	&+4478670640088860849920 \cos (28 \theta )+742545200580675148800 \cos (30 \theta ) \\
	&+105932607489264338640 \cos (32 \theta )+12901031043491036160 \cos (34 \theta ) \\
	&+1326253783709986560 \cos (36 \theta)+114403575099939840 \cos (38 \theta ) \\ &
	+8146060456248000 \cos (40 \theta )+456087964400640 \cos (42 \theta ) \\ &
	+20929545711360 \cos (44 \theta )+855236828160 \cos (46 \theta ) \\ & +18931023540 \cos (48 \theta ).
\end{align*}

\newpage

\subsubsection*{Case III: $\beta = 16, M = 6$}
\[
R_2(\theta) = \frac{30}{2\pi \cdot 2889253496242619386328267523990000}(2 \sin \theta)^{16} r(2 \cos \theta),
\]
where  
\begin{align*}
	r(y) &= 99561092450391000 \cdot y^{64}-2293887570057008640 \cdot y^{62} \\ & +29234092041020655360 \cdot y^{60}  -233037842068173542400
	\cdot y^{58} \\ & +2570110185308835312000 \cdot y^{56}-34769352212608261248000 \cdot y^{54} \\ &+370237231600199029946880
	\cdot y^{52}-2750986150525240158136320 \cdot y^{50} \\ & +15868045721917816108624800 \cdot y^{48}  -85719567550218211588492800
	\cdot y^{46} \\ & +489690642781322769058272000 \cdot y^{44}  -2661618272154068193895019520 \cdot y^{42} \\ & +12088358584235274923179678080
	\cdot y^{40}  -45133280783262574039660723200 \cdot y^{38} \\ & +149519236121248085045619494400 \cdot y^{36} -479337485462335081099964160000
	\cdot y^{34} \\ & +1468477085601996344193043055760 \cdot y^{32}  -3958521784145569339542873469440
	\cdot y^{30} \\ &+9037773641518524206215550496000 \cdot y^{28}  -18187118515733318049926150323200
	\cdot y^{26}\\ &+34248319980119855364931166390400 \cdot y^{24}  -59443896348939732246057042201600
	\cdot y^{22}\\ &+88617988648727371296927130867200 \cdot y^{20}  -111345609884277417307472304691200
	\cdot y^{18}\\ &+124447163190041591180410092163200 \cdot y^{16}  -125007449747844157477063579699200
	\cdot y^{14}\\ &+103281272904656629583781486105600 \cdot y^{12}  -68665200143241567896813963980800
	\cdot y^{10}\\ &+39413141819233148796281980070400 \cdot y^8  -16767860222869455568077756518400 \cdot y^6 \\ & +5318174506889516654964627302400
	\cdot y^4  -1088464869074174319545545728000 \cdot y^2\\ &+108656639093091455882121691800.
\end{align*}

\begin{figure}[htbp] 
	\centering
	\includegraphics[width=4in]{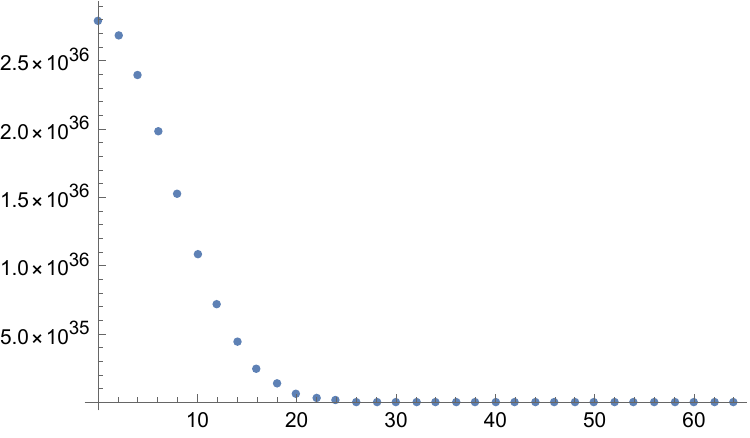}
	\caption{\label{fig:m6fourier} A plot of the non-zero Fourier coefficients
	of $r(2 \cos \theta)$ when $\beta=16$ and $M=6.$}
\end{figure}		

\newpage

It follows that
\begin{align*}
	r(2 \cos \theta) &= 1392968344952515316713424670628254600 \\ &+ 2683136499928597908237146479261286400 \cos (2 \theta ) \\
	&+2396855326278025276738716960654336000 \cos (4 \theta )\\ &+1985732466477010409334394895339520000 \cos (6 \theta) \\
	&+1525453313439224718086765162092185600 \cos (8 \theta )\\ &+1086333660616440626740600768519372800 \cos (10 \theta ) \\
	&+716913950607648252913961387875737600 \cos (12\theta )\\ &+438255383130013294869806055297024000 \cos (14 \theta ) \\
	&+248041177919425024623387656839224000 \cos (16 \theta )\\ &+129896828278260908317140114074419200 \cos(18 \theta ) \\
	&+62900066613343210359500666530713600 \cos (20 \theta )\\ &+28140648277386383160361255935590400 \cos (22 \theta ) \\
	&+11621272276750613056217285512550400 \cos(24 \theta )\\ &+4425519260407430865436223345049600 \cos (26 \theta ) \\
	&+1552230781514297298085961466777600 \cos (28 \theta )\\ &+500786818482709704188864880230400 \cos (30
	\theta ) \\
	&+148393515635861604354307960984800 \cos (32 \theta )\\ &+40319458262945776918999277568000 \cos (34 \theta ) \\
	&+10025440522668881880715791360000 \cos (36 \theta )\\ &+2276334643199562401445521326080 \cos (38 \theta ) \\
	&+470868639752793808225590359040 \cos (40 \theta )\\ &+88480327233922073304953978880 \cos (42 \theta)
	\\
	&+15048243843315362505350553600 \cos (44 \theta )\\ &+2307932656429825238645145600 \cos (46 \theta ) \\
	&+318247821460574863560331200 \cos (48 \theta)\\ &+39229136328273704545075200 \cos (50 \theta ) \\
	&+4272783290612194619289600 \cos (52 \theta )\\ &+407613854944772152934400 \cos (54 \theta ) \\
	&+34604318070329790182400 \cos  (56 \theta )\\ &+2662759282536706129920 \cos (58 \theta ) \\
	&+175456450154948751360 \cos (60 \theta )\\ &+8156044693536030720 \cos (62 \theta )\\ &+199122184900782000 \cos (64
	\theta )
\end{align*}

\newpage

\bibliography{bibliography}

\noindent\rule{4cm}{.5pt}
\vspace{.25cm}

\noindent {\sc \small Christopher D.~Sinclair}\\
{\small Department of Mathematics, University of Oregon, Eugene OR 97403} \\
email: {\tt csinclai@uoregon.edu}

\vspace{.25cm}

\noindent {\sc \small Jonathan M.~Wells}\\
{\small Department of Statistics, Grinnell College, Grinnell IA 50112} \\
email: {\tt wellsjon@grinnell.edu}

\end{document}